\documentclass[11pt]{article}

\usepackage{amssymb,amsmath,amsthm}
\usepackage{graphicx,url}

\newtheorem{theorem}{Theorem}
\newtheorem{lemma}[theorem]{Lemma}

\newcommand\R{\mathbb{R}}

\newcommand\E{\mathbb{E}}

\newcommand\Prb{\mathbb{P}}
\newcommand\cP{{\mathcal P}}

\newcommand\Th[1]{Theorem~\ref{t:#1}}

\begin{document}

\title{Unique coverage in Boolean models}
\author{M.~Haenggi \and A.~Sarkar}





\maketitle

\begin{abstract}

Consider a wireless cellular network consisting of small, densely scattered base stations.
A user $u$ is {\em uniquely covered} by a base station
$b$ if $u$ is the only user within distance $r$ of $b$. This makes it possible to assign the user $u$ to the
base station $b$ without interference from any other user $u'$. 
We investigate
the maximum possible proportion of users who are uniquely covered.
We solve this problem completely in one
dimension and provide bounds, approximations and simulation results for the two-dimensional case.

\end{abstract}


\section{Introduction}

Consider a wireless cellular network consisting of small, densely scattered base stations, each with limited 
processing capability. (In \cite{net:Caire16arxiv} and the related engineering literature, the small base stations 
are called {\em remote radio heads}.) In such a network, a user $u$ is {\em uniquely covered} by a base station 
$b$ if $u$ is the only user within distance $r$ of $b$. This makes it possible to assign the user $u$ to the
base station $b$ without interference from any other user $u'$. Ideally, we would like to assign a base station
to every user. However, the underlying stochastic geometry will prevent this. In this paper, we investigate 
the maximum possible proportion of users who can be uniquely assigned base stations, as the communication range 
$r$ varies, for each pair of densities of both users and base stations. 

Although we have just referred to {\em two} densities, only their ratio is significant; in other words,
the model can be scaled so that we expect one user per unit area. Accordingly, we set the intensity of users
to be one. Thus the only parameters we need to consider are the density $\mu$ of base stations, and the range $r$. 
Moreover, we note that our analysis also solves the problem, considered in~\cite{net:Caire16arxiv}, of uniquely assigning 
users to base stations (so as to avoid {\em pilot contamination}); to see this, simply interchange the roles of 
users and base stations.

All logarithms in this paper are to base $e$.

\section{Model}

Our model is as follows. Fix $r>0$, and let $\cP$ and $\cP'$ be independent Poisson processes, of
intensities $\mu$ and 1 respectively, in $\R^d$. The main case of interest is $d=2$. The points of
$\cP$ represent the base stations, and the points of $\cP'$ represent the users. A user $u\in\cP'$ is
uniquely covered by a base station $b\in\cP$ if firstly $||b-u||<r$, and secondly $||b-u'||\ge r$
for every other user $u'\in\cP'$. We wish to calculate (or estimate) the proportion $p^d(\mu,r)$ of
users who are uniquely covered by base stations; note that this proportion is also the probability that
an arbitrary user is uniquely covered by a base station.

\section{A general result}

In order to state our main result, we need some notation. First, for simplicity, we will initially consider
just the case $d=2$. Next, let $D=D(O,r)$ be the fixed open disc of radius
$r$, centered at the origin $O$. Write $f_r(t)$ for the probability density function of the fraction $t$ of $D$ which
is left uncovered when discs of radius $r$, whose centers are a unit intensity Poisson process, are placed in the entire
plane $\R^2$. There is in general no closed-form expression for $f_r(t)$; however, the function is easy to
estimate by simulation.

\begin{theorem}
In two dimensions, we have
\begin{equation}\label{basic}
p^2(\mu,r)=\int_0^1(1-e^{-\mu\pi r^2t})f_r(t)\,dt.
\end{equation}
\end{theorem}
\begin{proof}
The main idea of the proof is to put down the users first, and then, for a fixed user $u$, calculate
the probability that a base station $b$ ``lands" in such a way that $u$ is uniquely covered by $b$. To this end,
place a disc $D(u,r)$ of radius $r$ around each user $u$, and
then a fixed user $u$ is uniquely covered if there is a base station $b\in D(u,r)$ such that $b\not\in D(u',r)$
for all other users $u'\not= u$. Let $X$ be the random variable representing the uncovered area fraction
of $D(u,r)$ when all the other discs $D(u',r)$ are placed randomly in the plane. Then
\[
\Prb(u{\rm \ is\ covered}\mid X=t)=1-e^{-\mu\pi r^2t},
\]
since for $u$ to be covered we require that some base station $b$ lands in the uncovered region in $D(u,r)$,
which has area $\pi r^2t$. (Here, by ``uncovered", we mean ``uncovered by the union of all the other discs
$\bigcup_{u'\not= u}D(u',r)$".) Consequently,
\[
p^2(\mu,r)=\int_0^1 \Prb(u{\rm \ is\ covered}\mid X=t)f_r(t)\,dt=\int_0^1(1-e^{-\mu\pi r^2t})f_r(t)\,dt,
\]
as required.
\end{proof}

The same argument yields the following result for the general case. For $d\ge 1$, write $D^d(O,r)$ for the
$d$-dimensional ball of radius $r$ centered at the origin $O$, and $f^d_r(t)$ for the probability density function
of the fraction $t$ of $D^d(O,r)$ which is left uncovered when balls of radius $r$, whose centers are a unit intensity
Poisson process, are placed in $\R^d$. Finally, let $V_d$ be the volume of the unit-radius ball in $d$ dimensions.
\begin{theorem}
In $d$ dimensions, we have
\[
p^d(\mu,r)=\int_0^1(1-e^{-\mu V_d r^dt})f^d_r(t)\,dt.
\]
\end{theorem}

\section{The case $d=1$}

Unfortunately, $f^d_r(t)$ is only known exactly when $d=1$. The result is summarized in the following lemma,
in which for simplicity we consider the closely related function $g_r(s):=f^1_r(s/2r)$, which represents the total 
uncovered length in $(-r,r)$.

\begin{lemma}
In one dimension, we have
\[
g_r(s):=f^1_r(s/2r)=
\begin{cases}
1-e^{-2r}(1+2r)&\text{\rm point mass at $s=0$}\\
(2+2r-s)e^{-(2r+s)}&0<s<2r\\
e^{-4r}&\text{\rm point mass at $s=2r$}.\\
\end{cases}
\]
\end{lemma}
\begin{proof}
Consider the interval $I_r:=D^1(O,r)=(-r,r)$. The uncovered length $U$ of $I_r$ is determined solely by the
location of the closest user $u_l$ to the left of the origin $O$, and the closest user $u_r$ to the right of $O$. Suppose
indeed that $u_l$ is located at $-x$ and that $u_r$ is located at $y$. Then it is easy to see that if $x+y\le 2r$, we have $U=0$;
in other words, all of $I_r$ is covered by $D(u_l,r)\cup D(u_r,r)$ when $x+y\le 2r$. At the other extreme, if both $x\ge 2r$
and $y\ge 2r$, then $U=2r$; in this case the entire interval $I_r$ is left uncovered by $D(u_l,r)\cup D(u_r,r)$, and so by
the union $\bigcup_{u}D(u,r)$. In general, a lengthy but routine case analysis gives
\[
U=
\begin{cases}
0&x+y\le 2r\\
x+y-2r&x+y\ge 2r,x\le 2r,y\le 2r\\
x&0\le x\le 2r,y\ge 2r\\
y&0\le y\le 2r,x\ge 2r\\
2r&x\ge 2r,y\ge 2r.\\
\end{cases}
\]
This immediately yields the point masses of $g_r(s)$, since $x+y$ has a gamma distribution of mean 2, and $x$ and $y$ are each
exponentially distributed with mean 1. For $0<s<2r$ we find, using the above expression, that
\[
g_r(s)=2e^{-2r}\cdot e^{-s}+\int_{s}^{2r}e^{-x}e^{-(2r+s-x)}\,dx=(2+2r-s)e^{-(2r+s)},
\]
completing the proof of the lemma.
\end{proof}
\noindent Using this lemma, we obtain the following expression for $p^1(\mu,r)$.
\begin{theorem}
In one dimension, we have
\[
p^1(\mu,r)=\frac{\mu e^{-2r}(\mu+2r+2r\mu)-\mu^2e^{-2r(2+\mu)}}{(1+\mu)^2}.
\]
\end{theorem}
\begin{proof}
From Theorem 2 and Lemma 3 we have
\begin{align*}
p^1(\mu,r)&=\int_0^{2r}g_r(s)(1-e^{-\mu s})\,ds\\
&=e^{-4r}(1-e^{-2r\mu})+\int_0^{2r}(2+2r-s)e^{-(2r+s)}(1-e^{-\mu s})\,ds\\
&=\frac{\mu e^{-2r}(\mu+2r+2r\mu)-\mu^2e^{-2r(2+\mu)}}{(1+\mu)^2}.
\end{align*}
\end{proof}
\noindent 
$p^1(\mu,r)$ is illustrated in Fig.~\ref{fig:one_dim}.
The value of $r$ that maximizes $p^1$ is
\begin{equation}
   r_{\rm opt}(\mu)=\frac{1+\mathcal{W}(\mu(\mu+2)e^{-1})}{2\mu+2} ,
 \label{r_opt}
\end{equation}
where $\mathcal{W}$ is the (principal branch of the) Lambert W-function. It is easily seen
that $r_{\rm opt}(0)=1/2$ and that $r_{\rm opt}$ decreases with $\mu$.

\begin{figure}
\centerline{\includegraphics[width=8cm]{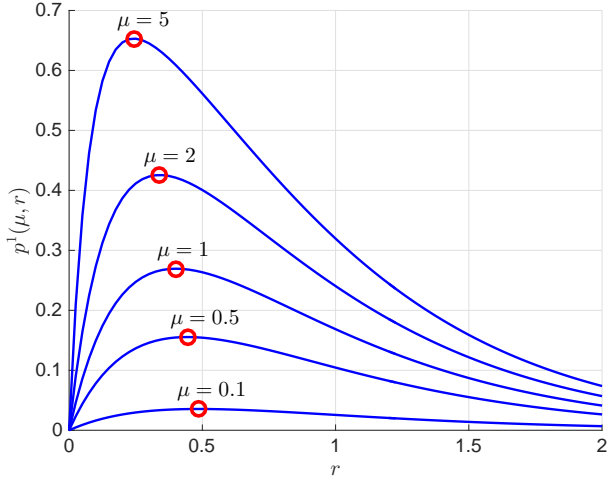}}
\caption{Fraction of users that are uniquely covered in one dimension. The circle indicates the maximum
$p^1(\mu,r_{\rm opt})$, where $r_{\rm opt}$ is given in \eqref{r_opt}.}
\label{fig:one_dim}
\end{figure}

\section{The case $d=2$}

In two dimensions, although the function $f^2_r(t)$ is currently unknown, it can be approximated by simulation, and then 
the integral \eqref{basic} can be computed numerically. While this still involves a simulation, it is more efficient than 
simulating the original model itself, since $f^2_r$ can be used to determine the unique coverage probability for many 
different densities $\mu$ (and the numerical evaluation of the expectation over $X$ is very efficient).
The resulting unique coverage probability $p^2(\mu,r)$ is illustrated in Fig.~\ref{fig:two_dim}.
The maxima of $p^2(\mu,r)$ over $r$, achieved at $p^2(\mu,r_{\rm opt}(\mu))$, are highlighted using circles.
Interestingly, $r_{\rm opt}(\mu)\approx 4/9$ for a wide range of values of $\mu$; the average of $r_{\rm opt}(\mu)$ 
over $\mu\in [0,10]$ appears to be about $0.45$.

The simulated $f_r(t)$ is shown in Fig.~\ref{fig:pdf_fr} for $r=3/9,4/9,5/9$. Remarkably, the density
$f_{4/9}(t)$ is very close to uniform (except for the point masses at $0$ and $1$). If the distribution were in fact uniform,
writing $v=\E(X)=e^{-\pi r^2}=e^{-16\pi/81}\approx 0.538$, we would have
\begin{equation}
\hat f_{4/9}(t)=\begin{cases}
1+v^4-2v\approx 0.008&\text{\rm point mass at $t=0$}\\
2(v-v^4)\approx 0.908 &0<t<1\\
v^4 \approx 0.084 &\text{\rm point mass at $t=1$}.\\
\end{cases}
\label{f49}
\end{equation}
Here, $v^4=e^{-4\pi r^2}$ is the probability that no other user is within distance $2r$, in which case the entire
disc $D(O,r)$ is available for base stations to cover $O$. The constant $2(v-v^4)$ is also shown in Fig.~\ref{fig:pdf_fr} 
(dashed line). Substituting \eqref{f49} in \eqref{basic} yields the following approximation to $p^2(\mu,4/9)$ and to 
$p^2(\mu,r_{\rm opt})$:
\begin{equation}
p^2(\mu,r_{\rm opt}(\mu))\approx 2(v-v^4)\left(1-\frac{1-e^{-c}}{c}\right)+v^4 (1-e^{-c}), \quad c=\mu \pi (4/9)^2=-\mu\log v.
\label{p2_approx}
\end{equation}
This approximation is shown in Fig.~\ref{fig:p2_opt}, together with the exact numerical result. For $\mu\in [3,7]$, 
the curves are indistinguishable.

For small $\mu$, $p^2(\mu,r)\approx e^{-\pi r^2}(1-e^{-\mu\pi r^2})$ (see \Th{p2_bound} immediately below), and so 
\[
r_{\rm opt}(\mu)\approx \sqrt{\frac{\log(1+\mu)}{\mu\pi}}\to \pi^{-1/2}
\] 
as $\mu\to 0$.

\begin{figure}
\centerline{\includegraphics[width=8cm]{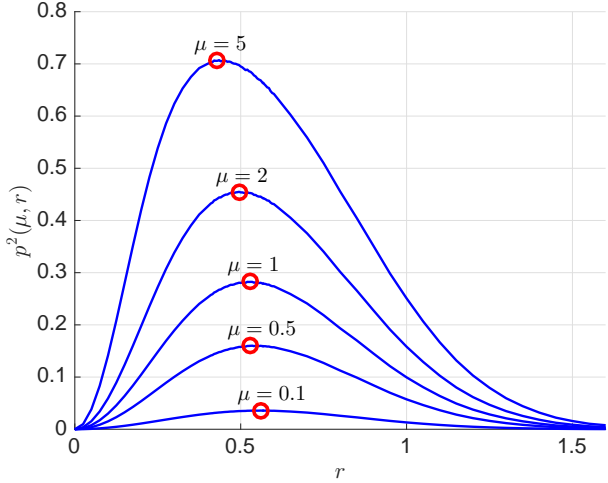}}
\caption{Fraction of users that are uniquely covered in two dimensions. The circles indicate the maxima
$p^2(\mu,r_{\rm opt})$.}
\label{fig:two_dim}
\end{figure}

\begin{figure}
\centerline{\includegraphics[width=8cm]{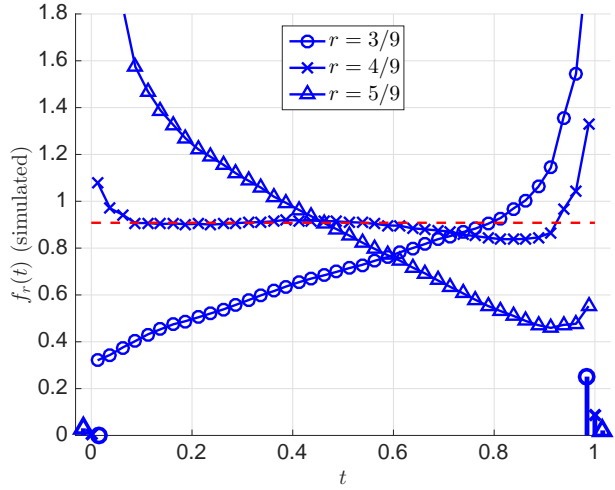}}
\caption{Simulated densities $f_r(t)$ for $r=3/9,4/9,5/9$ in two dimensions. The vertical lines near $0$ and $1$
indicate the point masses.
The dashed line is the uniform approximation \eqref{f49} for $r=4/9$.}
\label{fig:pdf_fr}
\end{figure}

\begin{figure}
\centerline{\includegraphics[width=8cm]{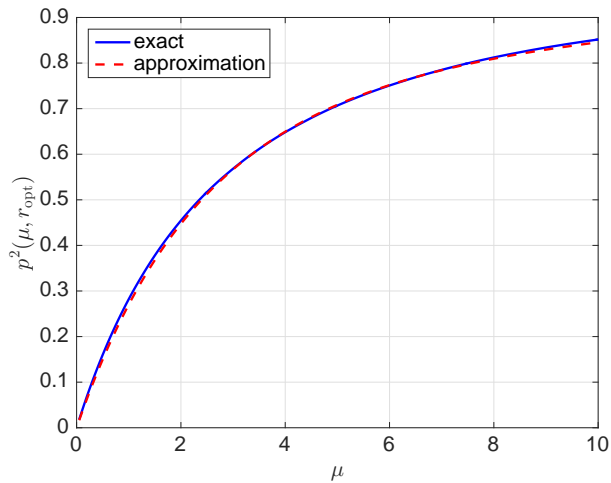}}
\caption{Maximum fraction of users that are uniquely covered in two dimensions.
The dashed line is the approximation in \eqref{p2_approx}.}
\label{fig:p2_opt}
\end{figure}

Next, we turn to bounds and approximations.
It is straightforward to obtain a simple lower bound for $p^2(\mu,r)$.
\begin{theorem}
\label{t:p2_bound}
$p^2(\mu,r)\ge e^{-\pi r^2}(1-e^{-\mu\pi r^2})$.
\end{theorem}
\begin{proof}
A given user is covered if there is a base station within distance $r$ (this event has probability $1-e^{-\mu\pi r^2}$), and
if there is no other user within distance $r$ of that base station (this event has probability $e^{-\pi r^2}$). These last
two events are independent.
\end{proof}
\noindent This bound should become tight as $\mu\to 0$ (with $r$ fixed), or as $r\to 0$ (with $\mu$ fixed), since, in those limiting
scenarios, if there is a base station within distance $r$ of a user, it is likely to be the only such base station.

Finally, here is an approximation for $p^2(\mu,r)$ when $r$ is large. We use standard asymptotic notation, so that $f(x)\sim g(x)$
as $x\to\infty$ means $f(x)/g(x)\to1$ as $x\to\infty$. In our case, we will have $r\to\infty$ with $\mu$ fixed.

\begin{theorem}
\label{t:p2_asymptote}
As $r\to\infty$ with $\mu$ fixed, $p^2(\mu,r)\sim \mu\pi r^2e^{-\pi r^2}$.
\end{theorem}
\begin{proof} (Sketch)
We recall Theorem 1, which states that
\[
p^2(\mu,r)=\int_0^1(1-e^{-\mu\pi r^2t})f_r(t)\,dt,
\]
and attempt to approximate $f_r(t)$ as $r\to\infty$.

To this end, it is convenient to describe the geometry of the union of discs $\bigcup_{u\in\cP'}D(u,r)$ in some detail. Such
coverage processes have been studied extensively in the mathematical literature~\cite{net:Gilbert65,net:Hall88,net:Janson86,net:Meester96}; our approach follows
that in~\cite{net:Balister09springer,net:Balister10aap}. The main idea is to consider the boundaries $\partial D(u,r)$ of the discs $D(u,r)$, rather than the
discs themselves. Consider a fixed disc boundary $\partial D(u,r)$. This boundary intersects the boundaries $\partial D(u',r)$
of all discs $D(u',r)$ whose centers $u'$ lie at distance less than $2r$ from $u$. There are an expected number $4\pi r^2$ of
such points $u'\in\cP'$, each contributing two intersection points $\partial D(u,r)\cap\partial D(u',r)$, and each intersection
is counted twice (once from $u$ and once from $u'$). Therefore we expect $4\pi r^2$ intersections of disc boundaries per unit area
over the entire plane; note that these intersections {\it do not} form a Poisson process, since they are constrained to lie on
various circles.

The next step is to move from intersections to regions. The disc boundaries partition the plane into small ``atomic" regions.
Drawing all the disc boundaries in the plane yields an infinite plane graph, each of whose vertices (disc boundary intersections)
has four curvilinear edges emanating from it. Each such edge is counted twice, once from each of its endvertices, so there are
almost exactly twice as many edges as vertices in any large region $R$. It follows from Euler's formula $V-E+F=2$ for plane
graphs~\cite{net:Bollobas98} that the number of atomic regions in $R$ is asymptotically the same as the number of intersection points in
$R$. Moreover, each vertex borders four atomic regions, so that the average number of vertices bordering an atomic region is
also four. Note that this last figure is just an average, and that many atomic regions will have less than, or more than, four
vertices on their boundaries.

The third step is to return to the discs themselves and calculate the expected number of {\it uncovered} atomic regions per unit area.
It is most convenient to calculate this in terms of uncovered intersection points. A fixed intersection point is uncovered by
$\bigcup_{u\in\cP'}D(u,r)$ with probability $e^{-\pi r^2}$ (using the independence of the Poisson process), so we expect
$4\pi r^2 e^{-\pi r^2}$ uncovered intersections, and so $\pi r^2 e^{-\pi r^2}$ uncovered regions, per unit area in $R$. Therefore
the expected number of uncovered regions in $D(u,r)$, which has area $\pi r^2$, is $\alpha=(\pi r^2)^2 e^{-\pi r^2}\to 0$.

How large are these uncovered atomic regions? To answer this, recall that the expected uncovered area in $D(u,r)$ is $\pi r^2 e^{-\pi r^2}$.
The uncovered atomic regions form an approximate Poisson process, so that the probability of seeing two uncovered regions in $D(u,r)$
is negligible. Now let $X_r$, with density function $f_r(t)$, be the uncovered area fraction in $D(u,r)$. We have $E(X_r)=e^{-\pi r^2}$,
but $\Prb(X_r=0)\sim e^{-\alpha}\sim 1-\alpha$. Writing now $Y_r$ for the expected uncovered area fraction in $D(u,r)$ conditioned on $X_r>0$,
and $h_r(t)$ for the density of $Y_r$, we see that $\E(Y_r)\sim\alpha^{-1}\E(X_r)=(\pi r^2)^{-2}$. In other words, if there is uncovered
area in $D(u,r)$, it occurs in one atomic region of expected area $(\pi r^2)^{-1}$. Consequently, we have
\begin{align*}
p^2(\mu,r)&=\int_0^1(1-e^{-\mu\pi r^2t})f_r(t)\,dt
\sim\alpha\int_0^1(1-e^{-\mu\pi r^2t})h_r(t)\,dt\\
&\sim\alpha\mu\pi r^2\int_0^1th_r(t)\,dt
=\alpha\mu\pi r^2\E(Y_r)
\sim\alpha\mu(\pi r^2)^{-1}
=\mu\pi r^2e^{-\pi r^2}.
\end{align*}
\end{proof}

Note that this is the same result that we would have obtained from the incorrect argument that $X_r$ is concentrated around its mean,
whereas in fact its density $f_r(t)$ has a large point mass at $t=0$. Indeed, the thrust of the above argument is that, for the relevant
range of $t$ (namely, for $t=O((\pi r^2)^{-2})$), $1-e^{\mu\pi r^2t}-\mu\pi r^2 t=O(r^4t^2)=O(r^{-4})$, which is asymptotically negligible
compared to the remaining terms.

Fig.~\ref{fig:p2_bounds} shows $p^2(\mu,r)$, together with the lower bound from \Th{p2_bound} and the asymptote
from \Th{p2_asymptote}. As predicted, \Th{p2_bound} is close to the truth when $r$ is small, while \Th{p2_asymptote} is more accurate
for large values of $r$.

\begin{figure}[htp!]
\centering
\begin{minipage}{0.49\textwidth}
\includegraphics[width=\textwidth]{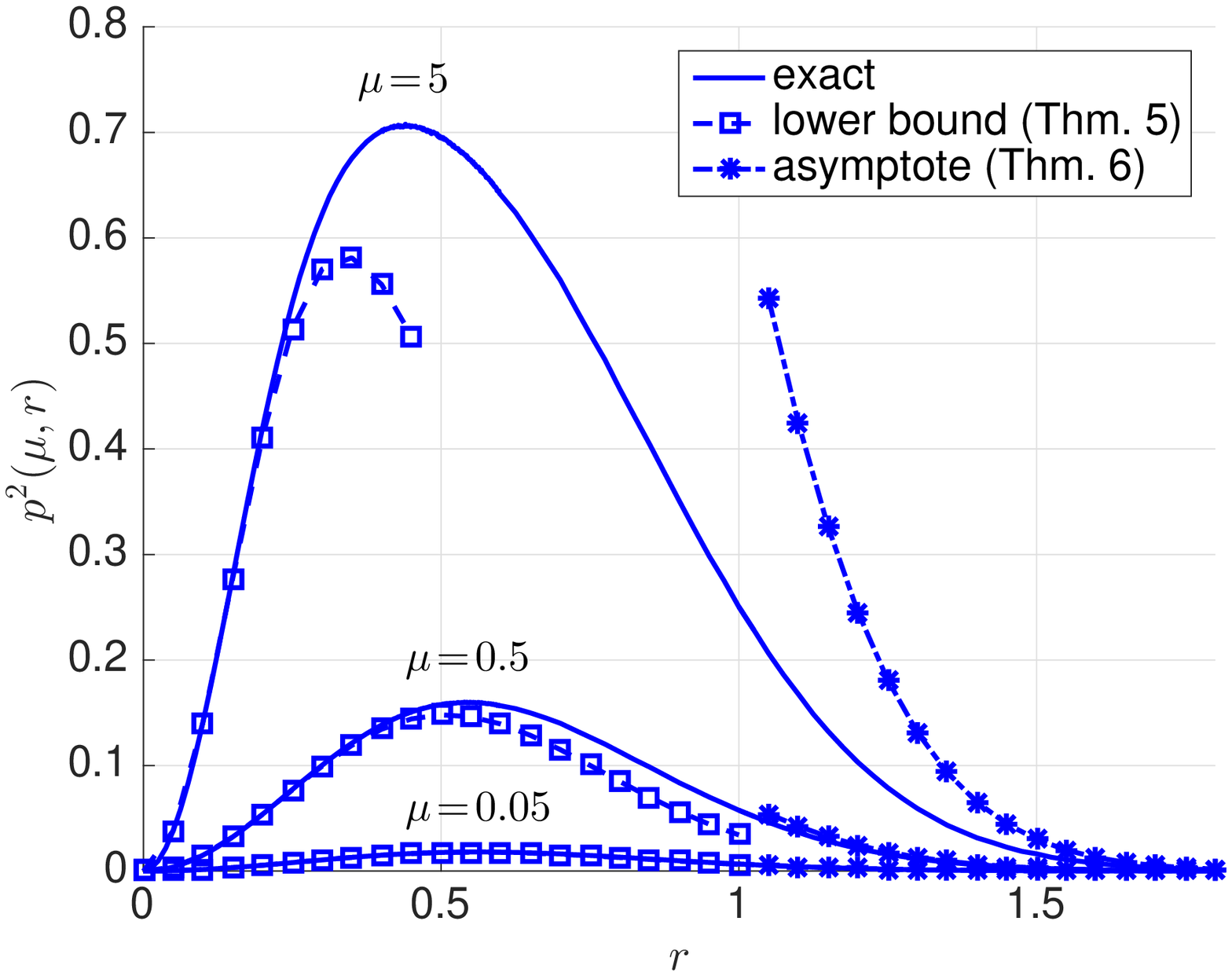}
\end{minipage}
\hfill
\begin{minipage}{0.49\textwidth}
\includegraphics[width=\textwidth]{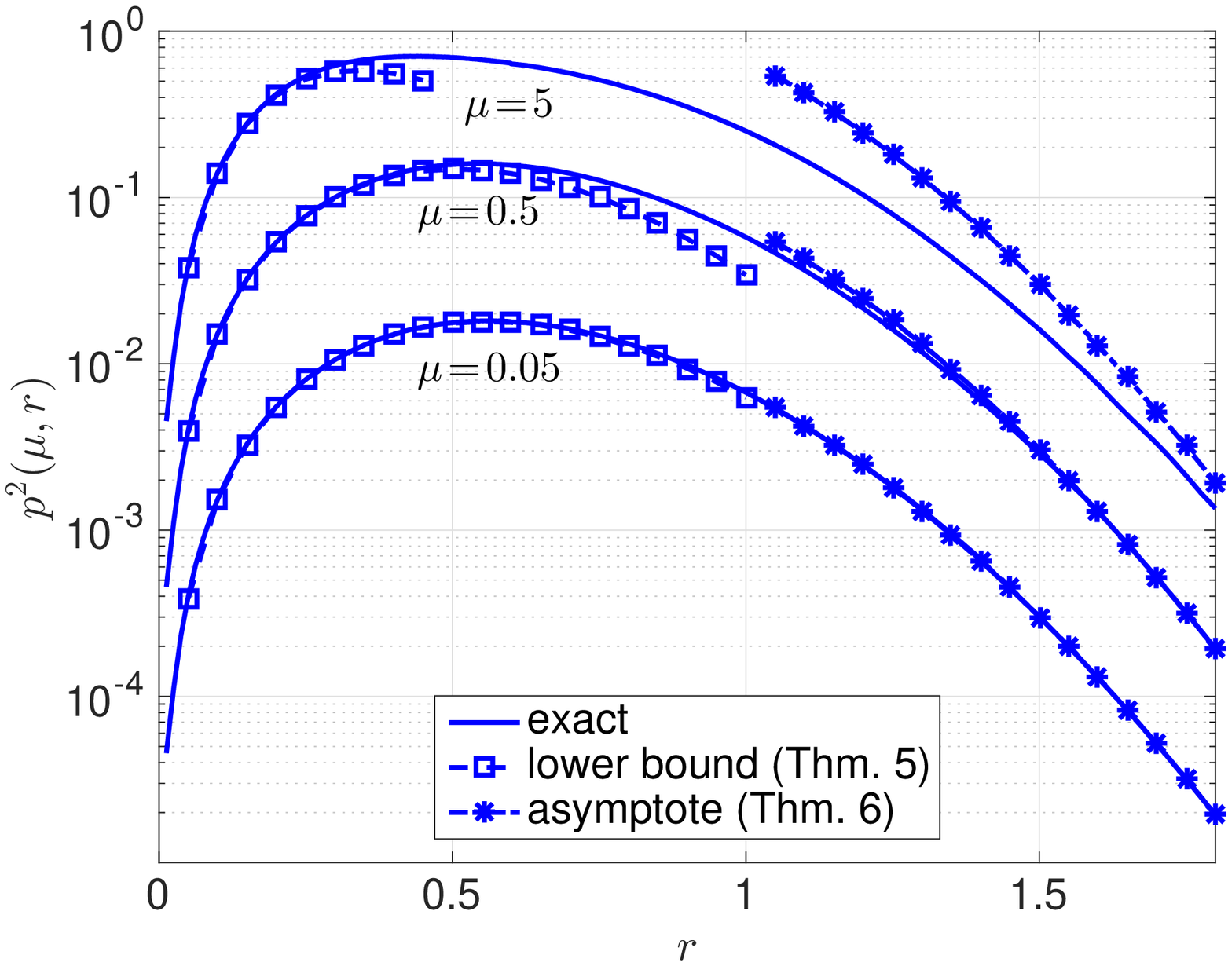}
\end{minipage}
\caption{$p^2(\mu,r)$ for $\mu=0.05,0.5,5$, with the lower bound from \Th{p2_bound} and the approximation from
\Th{p2_asymptote}. (Left) linear scale. (Right) logarithmic scale.}
\label{fig:p2_bounds}
\end{figure}

Both these last two results generalize to the $d$-dimensional setting in the obvious way; for simplicity we omit the details.

\section{Conclusions}

In this paper, we have investigated a natural stochastic coverage model, inspired by wireless cellular networks. For this model,
we have studied the maximum possible proportion of users who can be uniquely assigned base stations, as a function of the base station
density $\mu$ and the communication range $r$. We have solved this problem completely in one dimension
and provided bounds, approximations 
and simulation results for the two-dimensional case. We hope that our work will stimulate further research in this area.

\section{Acknowledgements}

We thank Giuseppe Caire for bringing this problem to our attention.
This work was supported by the US National Science Foundation
[grant CCF 1525904].


\bibliographystyle{ieeetr}

\end{document}